\documentclass[11pt, letterpaper]{article} 
\usepackage{amsmath,amssymb,fullpage}

\newtheorem{theorem}{Theorem}

\newtheorem{conjecture}[theorem]{Conjecture}
\newtheorem{lemma}{Lemma}
\newtheorem{definition}[theorem]{Definition}

\def\FullBox{\hbox{\vrule width 8pt height 8pt depth 0pt}}

\newcommand{\myqed}{\;\;\;\FullBox}

\newcommand{\bbE}{\mathbb{E}}
\newcommand{\bbN}{\mathbb{N}}

\newenvironment{proof}{\noindent{\bf Proof:~~}}{\(\myqed\)}
\newenvironment{proof-of-lemma}{{\em Proof of Lemma:~~}}{\(\myqed\)}
\newenvironment{proof-of-lemma-noqed}{{\em Proof of Lemma:~~}}{\(\)}

\title{The Price of Uncertain Priors in Source Coding\thanks{To appear in IEEE
Transactions on Information Theory. This paper was 
presented in part at Allerton 2015. M.~Braverman is partially supported by NSF 
Award CCF-1525342, NSF CAREER award CCF-1149888, a Packard Fellowship in Science
and Engineering, and the Simons Collaboration on Algorithms and Geometry.
B.~Juba was partially supported by ONR grant number N000141210358, an AFOSR 
Young Investigator Award, and NSF award CCF-1718380.}}
\author{Mark Braverman\thanks{M.~Braverman is with the Department of Computer
Science, Princeton University, Princeton, NJ 08540, USA (E-mail: mbraverm@cs.princeton.edu)}
\and Brendan Juba\thanks{B.~Juba is with the Department of Computer Science and
Engineering, Washington University, St.~Louis, MO 63130, USA (E-mail: 
bjuba@wustl.edu). Portions of this work were performed while B. Juba was 
affiliated with the School of Engineering and Applied Science, Harvard 
University, Cambridge, MA 02138, USA.}
\thanks{\copyright~2017 IEEE. Personal use of this material is permitted. However, permission to use this material for any other purposes must be obtained from the IEEE by sending a request to pubs-permissions@ieee.org.}}

\begin{document}
\maketitle

\begin{abstract}
We consider the problem of one-way communication when the recipient does not
know exactly the distribution that the messages are drawn from, but has a {\em
``prior''} distribution that is known to be close to the source distribution, a 
problem first considered by Juba et al. We consider the question 
of how much longer the messages need to be in order to cope with the uncertainty
about the receiver's prior and the source distribution, respectively, as 
compared to the standard source coding problem. We consider two variants of this
uncertain priors problem: the original setting of Juba et al. in 
which the receiver is required to correctly recover the message with 
probability~1, and a setting introduced by Haramaty and Sudan, in 
which the receiver is permitted to fail with some probability $\epsilon$. In 
both settings, we obtain lower bounds that are tight up to logarithmically
smaller terms. In the latter setting, we furthermore present a variant of the 
coding scheme of Juba et al. with an overhead of $\log\alpha+\log 1/\epsilon+1$ 
bits, thus also establishing the nearly tight upper bound.
\end{abstract}

\section{Introduction}
In a seminal work, Shannon~\cite{shannon48} considered the 
problem of how to encode a message so that it can be transmitted and decoded 
reliably across a channel that introduces errors. Shannon's contribution in that
work was two-fold: first, he identified how large any encoding of messages would
need to be in the absence of noise -- the {\em ``source coding''} problem -- and
then identified the optimal length for encodings that can be decoded in spite of
errors introduced by the channel. The difference between these two lengths --
the number of extra, {\em redundant} bits from the standpoint of source coding 
-- may be viewed as the {\em ``price'' of noise-resilience}. 

Such work in information theory has all but settled the basic quantitative 
questions of noise-tolerant transmission in telecommunications. In natural 
communication, however, errors frequently arise not due to any kind of 
interference, but instead due to a lack of shared context. In the interest of 
understanding why natural language is structured so that such errors can occur 
and how they might be addressed, Juba et al.~\cite{jkks11} introduced a model of
{\em communication with uncertain priors.} This is a variant of the source 
coding problem in which the sender and receiver do not agree on the source 
distribution that the messages are drawn from. Thus, errors arise because the 
sender and receiver do not agree on which messages should be considered more 
likely, and should therefore receive shorter codewords so as to minimize the 
expected encoding length. We note that this problem also has applications to 
adaptive data compression, in which parties use their empirical observations of 
message frequencies to encode messages. This would be useful since the 
distribution over messages generally changes over time for a variety of reasons;
this clearly occurs in natural language content, for example, as new words are 
introduced and old ones fall out of use. Since different parties on a network 
will in general observe different empirical distributions of messages, they must
tolerate some (limited) inconsistency about the relative frequency of the 
different messages.

In the model of Juba et al., ``uncertainty'' about the priors is captured by the
following kind of distance between the priors used by the sender and receiver. 
Namely, if the sender has a source distribution $P$ and the receiver expects a 
distribution $Q$, then we say that $P$ and $Q$ are {\em $\alpha$-close} when 
$\frac{1}{\alpha}Q(x)\leq P(x)\leq \alpha Q(x)$ for every message $x$. Juba et 
al. then presented a scheme (building on the coding technique of Braverman and
Rao~\cite{br14}) in which every source $P$ can be encoded by $H(P)+2\log\alpha+
O(1)$ bits so that every decoder using an $\alpha$-close prior $Q$ will recover 
the message correctly. Thus, the Juba et al. scheme uses $2\log\alpha+O(1)$ bits
beyond the $H(P)$ bits achieved by standard solutions to the basic source coding
problem. 

Juba et al. had noted that $H(P)+\log\alpha$ bits is necessary for such 
``uncertain priors'' coding (at least for a prefix-coding variant), and asked 
whether the redundancy could be reduced to this $\log\alpha$ lower bound. In 
this work, we address this question by showing that it cannot. Indeed, in the 
original error-free coding setting, we show that the redundancy must be at least
$2\log\alpha$ up to terms of size $O(\log\log\alpha)$, and hence the {\em 
``price of uncertainty''} is, up to lower-order terms, an additional 
$2\log\alpha$ bits. We also consider the variant of the problem introduced by 
Haramaty and Sudan~\cite{hs16} in which the decoding is allowed to fail with 
some positive probability $\epsilon$. We also nearly identify the price of 
uncertainty in this setting: we note that the scheme of Juba et 
al./Braverman-Rao can be modified to give an uncertain priors coding of length 
$H(P)+\log\alpha+\log 1/\epsilon+1$, and show a lower bound on the redundancy of
$\log\alpha+\log 1/\epsilon-O(\log\log\alpha)$ when $\epsilon\geq 1/\alpha$. The
price of uncertainty $\alpha$ when error $\epsilon>1/\alpha$ is allowed is thus 
essentially reduced to $\log\alpha+\log 1/\epsilon$ bits.

We obtain our results by reducing a one-way communication complexity problem to 
uncertain prior coding: Consider the problem where Alice receives as input a 
message from a domain of size roughly $\alpha^2$, Bob receives as input a set 
$S$ of size roughly $\alpha$ containing Alice's message, and Alice's task is to
send the message to Bob. For our reduction, we note that there is a low-entropy 
distribution $P$ with most of its mass on Alice's input and an $\alpha$-close 
family of distributions corresponding to the sets $S$ that essentially capture 
this problem. Thus, a lower bound for the communication complexity problem 
translates directly to the desired lower bound on the redundancy since the 
entropy of $P$ is negligible compared to the overhead. Our lower bounds for
the two variants, error-free and positive error, of the uncertain priors coding 
problem are then obtained from lower bounds for this same problem in the 
analogous variant of the one-way communication complexity model. Specifically, 
we obtain a $2\log\alpha-O(\log\log\alpha)$ lower bound in the error-free model 
and a $\log\alpha+\log 1/\epsilon-O(\log\log\alpha)$ lower bound in the model 
with $\epsilon$ error, yielding our main results.

\subsection{Aside: why not the relative entropy/KL-divergence?}
A common misconception upon first learning about the model of Juba et al.~\cite
{jkks11} is that (a) the problem had already been solved and (b) the correct
overhead is given by the {\em relative entropy} (or ``{\em KL-divergence}''),
$RE(P\|Q)=\sum_xP(x)\log\frac{P(x)}{Q(x)}\leq\log\alpha$. Of course, our lower
bounds imply that this is incorrect, but it is useful to understand the reason.
Indeed, our problem is somewhat unusual in that the relative entropy is 
essentially the correct answer to a few similar problems, including (i) the 
problem where the {\em sender} does not know the source distribution $P$, only 
an approximation $Q$ and (ii) the problem where the communication is two-way 
(and the receiver can tell the sender when to stop)---this variant essentially 
follows from the work of Braverman and Rao~\cite{br14}. The difference is that 
in our setting unlike (i), the sender does not know $Q$ and unlike (ii), has no 
way to learn anything about it, apart from the fact that it is $\alpha$-close to
$P$. The sender's message must simultaneously address decoding by all possible 
$\alpha$-close priors using only this knowledge, hence the connection to a
worst-case communication complexity set-up.

We also note that the problem we consider is not addressed by the {\em universal
compression} schemes of Lempel and Ziv~\cite{lz77}. Lempel and Ziv's compression
schemes provide {\em asymptotically} good compression of many messages from, 
e.g., a Markovian source. By contrast, the problem we consider here concerns the
compression of a {\em single message.} The question of compressing multiple
messages is still interesting, and we will return to it in Section~\ref
{future-sec}.

\section{The Model and Prior State of the Art}

We now recall the model we consider in more detail and review the existing work
on this model. Our work concerns the uncertain priors coding problem, originally
introduced by Juba et al~\cite{jkks11}. In the following we will let $M$ denote
a set of {\em messages} and $\Delta(M)$ denote the probability distributions on
$M$. 
\begin{definition}[$\alpha$-close~\cite{jkks11}]
We will say that a pair of distributions $P,Q\in\Delta(M)$ are
{\em $\alpha$-close} for $\alpha\geq 1$ if for every $m\in M$,
\[
\frac{1}{\alpha}Q(m)\leq P(m)\leq \alpha Q(m).
\]
\end{definition}

In uncertain priors coding, one party ({\em ``Alice''}) wishes to send a message
drawn from a source distribution $P\in\Delta(M)$ using a one-way (noiseless) 
binary channel to another party ({\em ``Bob''}) who does not know $P$ exactly. 
Bob does know a distribution $Q\in\Delta(M)$, however, that is guaranteed to be 
$\alpha$-close to $P$, where Alice in particular knows $\alpha$. We assume that
Alice and Bob share access to an infinitely long common random string $R$. Our 
objective is to design an encoding scheme that, regardless of the pair of {\em 
``prior''} distributions $P$ and $Q$, enables Bob to successfully decode the 
message using as short a transmission as possible:
\begin{definition}[Error-free uncertain priors coding~\cite{jkks11}]
An {\em error-free uncertain priors coding scheme} is given by a pair of 
functions $E:M\times\bbN\times\{0,1\}^{\bbN}\times\Delta(M)\to\{0,1\}^*$ and 
$D:\{0,1\}^*\times\bbN\times\{0,1\}^{\bbN}\times\Delta(M)\to M$ such that for 
every $m\in M$, $\alpha\in\bbN$, $R\in\{0,1\}^{\bbN}$, $P\in\Delta(M)$, and 
$\alpha$-close $Q\in\Delta(M)$, if $c=E(m,\alpha,R,P)$, then $D(c,\alpha,R,Q)=
m$. When $R\in\{0,1\}^{\bbN}$ is chosen uniformly at random and $m$ is chosen
from $P$, we will refer to $\bbE_{m,R}[|E(m,\alpha,R,P)|]$ as the {\em encoding
length} of the scheme $(E,D)$ for $P$ and $\alpha$.
\end{definition}

The key quantity we focus on in this work is a measure of the overhead 
introduced by uncertain priors coding, as compared to standard source coding
where $P=Q$. 

\begin{definition}[Redundancy]
For $P\in\Delta(M)$, let $\ell(P)$ denote the optimal one-to-one encoding length
for $P$, i.e., if $C$ is the set of one-to-one maps from $M$ to the binary 
strings, $\ell(P)=\min_{c\in C}\bbE_{m}[|c(m)|]$ where the expectation is over
$m$ drawn from $P$ and $|c(m)|$ denotes the length of the encoding $c(m)$.
The {\em redundancy} of an uncertain priors coding scheme $(E,D)$ for a given
$\alpha\in\bbN$ is given by the maximum over $P\in\Delta(M)$ of the encoding 
length of $(E,D)$ for $\alpha$ and $P$ minus the optimal one-to-one encoding
length for $P$. That is, the redundancy is
\[
\max_{P\in\Delta(M)}\left\{\bbE_{m,R}[|E(m,\alpha,R,P)|]-\ell(P)\right\}
\]
where $m$ is drawn from $P$ and $R$ is chosen uniformly at random in the
expectation.
\end{definition}

As an uncertain priors coding scheme gives a means to perform standard source
coding (again, by taking $P=Q$, even when $\alpha=1$) it follows that the 
redundancy is always nonnegative.

Huffman codes~\cite{huffman52}, for example give a means to encode a message in 
this setting using on average no more than one more bit than the {\em entropy} 
of $P$, $H(P)=\sum_{m\in M}P(m)\log\frac{1}{P(m)}$. It is well known that this 
coding length is essentially optimal if we require $P$ to be a {\em prefix code}
(i.e., self-delimiting). But, it is possible to achieve a slightly better coding
length when the end of the message is marked for us. Elias~\cite{elias75} 
observed that the entropy is an upper bound for the cost of this problem as well
(which he credits to Wyner~\cite{wyner72}), whereas Alon and Orlitsky~\cite
{ao94} showed that the entropy of $P$ exceeds the optimal one-to-one encoding 
length $\ell(P)$ by at most $\log(\ell(P)+1)+\log e$. Furthermore, they observed
that this bound is approached by a geometric distribution as the parameter 
approaches zero (and it is thus essentially tight).

Juba et al.~\cite{jkks11}, using a coding technique introduced by Braverman and
Rao~\cite{br14}, exhibited an error-free uncertain priors coding scheme that 
achieved a coding length of the entropy plus a function of $\alpha$:
\begin{theorem}[Juba et al.~\cite{jkks11}]
\label{errorless-coding-thm}
There is an error-free uncertain priors coding scheme that achieves encoding
length at most $H(P)+2\log\alpha+2$.
\end{theorem}
Roughly, the scheme achieving this bound proceeds as follows. The shared random
string is interpreted as specifying, for each possible message $m$, an infinite
sequence of independent random bits. Each prefix of the string for $m$ is taken 
as a possible encoding of $m$: that is, to encode a message $m$, Alice computes
some sufficiently large index $i$, and transmits the first $i$ bits of the
common random string associated with $m$ to Bob. Specifically, if Bob's prior
is known to be $\alpha$-close, Alice chooses $i$ so that no other message $m'$ 
that shares an $i$-bit encoding with $m$ has prior probability $P(m)/\alpha^2$ 
or greater under $P$. It then follows immediately that since Bob's prior $Q$ is
$\alpha$-close to $P$, $m$ is the unique message of maximum likelihood under $Q$
with the given $i$-bit encoding. So, it suffices for Bob to simply output this
maximum-likelihood message.

Juba et al. also noted that a lower bound on the redundancy of $\log\alpha$ is 
easy to obtain (in a prefix-coding model). This follows essentially by taking a 
distribution over possible source distributions $P$ that are $\alpha$-close to a
common distribution $Q$, and noting that the resulting distribution over 
messages has entropy $H(P)+\log\alpha-o(1)$. Our first main theorem, in 
Section~\ref{errorless-lb-sec}, will improve this lower bound to $2\log\alpha-
O(\log\log\alpha)$, so that it nearly matches the upper bound given by 
Theorem~\ref{errorless-coding-thm} (up to lower-order terms).

We also consider a variant of the original uncertain priors coding problem,
introduced by Haramaty and Sudan~\cite{hs16}, in which we allow an error in
communication with some bounded but positive probability:
\begin{definition}[Positive-error uncertain priors coding]
For any $\epsilon>0$, an {\em $\epsilon$-error uncertain priors coding scheme} 
is given by a pair of functions $E:M\times\bbN\times\{0,1\}^{\bbN}\times
\Delta(M)\to\{0,1\}^*$ and $D:\{0,1\}^*\times\bbN\times\{0,1\}^{\bbN}\times
\Delta(M)\to M$ such that for every $\alpha\in\bbN$, $P\in\Delta(M)$, and 
$\alpha$-close $Q\in\Delta(M)$, when $m\in M$ is chosen according to $P$ and $R
\in\{0,1\}^{\bbN}$ is chosen uniformly at random,
\[
\Pr_{m,R}[D(E(m,\alpha,R,P),\alpha,R,Q) = m]\geq 1-\epsilon.
\]
Again, we will refer to $\bbE_{m,R}[|E(m,\alpha,R,P)|]$ as the {\em encoding 
length} of the scheme $(E,D)$ for $P$ and $\alpha$.
\end{definition}

\noindent
The redundancy for positive-error uncertain priors coding is then defined in
exactly the same way as for error-free coding.

We briefly note that the definition of Haramaty and Sudan~\cite{hs16} differs in
two basic ways. First, their definition does not include a common random string 
since they were primarily interested in deterministic coding schemes. Second,
they required that the decoder output a special symbol $\bot$ when it makes an
error. Our one positive result (Theorem~\ref{poserror-ub}) can be easily 
modified to satisfy this condition, but we prove our lower bound, Theorem~\ref
{poserror-lb}, for the slightly more lenient model stated here. 

We also briefly note that Canonne et al.~\cite{cgms14} have considered another 
variant of the basic model in which Alice and Bob do not share the common random
string $R$ perfectly, only correlated random strings. In both this imperfect 
randomness model and the deterministic model of Haramaty and Sudan, the known 
encoding schemes feature substantially greater redundancy than $2\log
\alpha$---the redundancy for these schemes is linear in the entropy and, in the 
case of the deterministic schemes of Haramaty and Sudan, furthermore depends on 
the size of $M$. It is an interesting open question for future work whether or
not lower bounds of this form can be proved for these other settings.

\section{The Price of Uncertainty}
We now establish lower bounds on the redundancy for uncertain priors coding
schemes, in both the error-free and positive error variants. We will see that 
both of these lower bounds are tight up to lower-order terms. Hence, at least
to the first order, we identify the ``price'' incurred by uncertainty about
a recipient's prior distribution, beyond what is inherently necessary for
successful communication in the absence of such uncertainty.

In both cases, our lower bounds are proved by exploiting the worst-case nature
of the guarantee over priors $P$ and $Q$ to embed a (worst-case) one-way
communication complexity problem into the uncertain priors coding problem. We
then essentially analyze both error-free and positive-error variants of the
communication complexity problem to obtain our lower bounds for the respective
uncertain priors coding problems.

\subsection{Lower bound for error-free communication}\label{errorless-lb-sec}

We first consider the original error-free variant of uncertain priors coding,
considered by Juba et al.~\cite{jkks11}. Earlier, in Theorem~\ref
{errorless-coding-thm}, we recalled that they gave an error-free uncertain
priors coding scheme with coding length $H(P)+2\log\alpha+2$. We now show that 
this is nearly tight, up to lower-order terms:

\begin{theorem}\label{errorless-lb}
For every error-free compression scheme for uncertain priors and every 
sufficiently large $\alpha$, there exists a pair of $\alpha$-close priors for 
which the scheme suffers redundancy at least $2\log\alpha-3\log\log\alpha-O(1)$.
\end{theorem}
\begin{proof}
Given $\alpha$, we will first define a family of priors for the sender and 
receiver that will contain a pair of priors with high redundancy for every 
scheme. Let $k$ be the largest integer such that $k\sqrt{\log k}\leq\alpha$ (so 
$\log\alpha=\log k+\frac{1}{2}\log\log k+O(1)$) and consider a message set of
size $k^2+1$. The family of priors is parameterized by a distinguished message
$m$ and a set $S$ of $k+1$ messages that includes $m$.

For a given $m$ and $S$, Alice's prior $P$ gives message $m$ probability $1-
1/\log k$ and gives the other $k^2$ messages probability $1/k^2\log k$. Bob's
prior $Q$, on the other hand, gives each message in $S$ probability $1/k\sqrt
{\log k}$, and gives the other $k^2-k$ messages uniform probability
$\frac{1}{k^2-k}\left(1-\frac{k+1}{k\sqrt{\log k}}\right)$.
\begin{lemma}\label{pqclose-lem}
For sufficiently large $\alpha$ and $k$, every pair $(P,Q)$ in the family are 
$\alpha$-close.
\end{lemma}
\begin{proof-of-lemma}
We chose $k$ so that $k\sqrt{\log k}\leq\alpha$. 
%it suffices to show that $P$ and $Q$ are $k\sqrt{\log k}$-close. 
Since the messages in $S$ have probability $1/k\sqrt{\log k}$ under $Q$, $m$ 
(which has probability between $1$ and $1/k\sqrt{\log k}$ under $P$) is 
$\alpha$-close, and we note that the rest of the messages in $S$
have probability $1/k^2\log k=\frac{1}{k\sqrt{\log k}}\frac{1}{k\sqrt{\log k}}
\geq \frac{1}{\alpha}\frac{1}{k\sqrt{\log k}}$ under $P$, so the rest of these 
messages in $S$ satisfy $\frac{1}{\alpha}Q(x)\leq P(x)\leq Q(x)<\alpha Q(x)$. 
The messages outside $S$ also have probability $1/k^2\log k$ under $P$, which is
less than $\frac{1}{k^2-k}\left(1-\frac{k+1}{k\sqrt{\log k}}\right)$ and 
certainly greater than $\frac{1}{(k^2-k)k\sqrt{\log k}}\left(1-\frac{k+1}{k\sqrt
{\log k}}\right)$ for sufficiently large $k$. Thus, actually, for such $x$ 
outside $S$, again $\frac{1}{\alpha}Q(x)\leq P(x)\leq Q(x)<\alpha Q(x)$.
\end{proof-of-lemma}

\begin{lemma}\label{hp-lem}
Every prior $P$ in the family has entropy $O(1)$ (and hence $\ell(P)$ is $O(1)$
as well).
\end{lemma}
\begin{proof-of-lemma-noqed}
With probability $1-\frac{1}{\log k}$, a draw from $P$ gives the message $m$ 
with self-information $-\log(1-1/\log k)\leq\frac{1}{\log k}$, and gives each of
the $k^2$ messages with self-information $\log(k^2\log k)=2\log k+\log\log k$ 
with probability $\frac{1}{k^2\log k}$. Thus, overall
\begin{align*}
H(P)&\leq (1-\frac{1}{\log k})\frac{1}{\log k}+k^2\frac{1}{k^2\log k}(2\log k+
\log\log k)\\
&\leq \frac{1+\log\log k}{\log k}+2.\ \myqed
\end{align*}
\end{proof-of-lemma-noqed}

It therefore suffices to give a lower bound on the expected codeword length to
obtain a lower bound on the redundancy up to an additive $O(1)$ term, since the
optimal encodings of the messages themselves only contribute $O(1)$ bits.
\begin{lemma}
Any error-free coding scheme must have expected codeword length at least
$2\log k-2\log\log k$ for some $\alpha$-close pair $P$ and $Q$ as described
above.
\end{lemma}
\begin{proof-of-lemma}
We note that by the min-max principle, it suffices to consider deterministic
coding schemes for the case where $m$ and $S$ are chosen uniformly at random
from the common domain of size $k^2+1$. In slightly more detail, we are 
considering a zero-sum game between an ``input'' player and a ``coding scheme'' 
player, in which the pure strategies for the input player are pairs of priors 
$P$ and $Q$ from our family, and the pure strategies for the coding player are 
error-free deterministic coding schemes. The payoff for the input player is the
expected encoding length of the chosen coding scheme on the chosen pair of 
priors. Randomized coding schemes with shared random strings can be viewed as a 
random choice of a deterministic scheme; Von Neumann's Min-max Theorem then
guarantees that the expected encoding length of the best randomized coding 
scheme is equal to the expected encoding length of the best deterministic coding
scheme for an optimal (hardest) distribution over priors. So, by showing that
for {\em some} concrete choice of distribution over priors the expected encoding
length must be at least $2\log k-2\log\log k$, we can infer that for every
randomized coding scheme there must exist a fixed choice of $m$ and $S$ under
which the expected encoding length is at least $2\log k-2\log\log k$, which
will complete the proof of the lemma.

So, suppose for contradiction that the expected codeword length of some
deterministic error-free coding scheme is less than $2\log k-2\log\log k$ when
the parameters $m$ and $S$ defining $P$ and $Q$ are chosen uniformly at random. 
Markov's inequality then guarantees that there is a ``collision''---an ambiguous
codeword for two distinct messages: in more detail, since the probability of
obtaining a codeword of length at least $2\log k-\log\log k$ is at most
\[
\frac{2-2(\log\log k)/\log k}{2-(\log\log k)/\log k}=1-\frac{1}{2\frac{\log k}
{\log\log k}-1}
\]
we find that with probability at least $1\left/\left(2\frac{\log k}{\log\log k}
-1\right)\right.$, the code length is at most $2\log k-\log\log k$. That is,
recalling that we have a uniform distribution over the domain, at least
$\frac{k^2+1}{2\frac{\log k}{\log\log k}-1}$ messages have such short codes.
But, we know that any unique code for so many messages has a codeword of
length at least
\[
\log(k^2+1)-\log\left(2\frac{\log k}{\log\log k}-1\right)-1>2\log k-\log\log k
\]
for sufficiently large $k$. Thus, these messages with short codes cannot have
been uniquely encoded.

So, there must be two messages $m_1$ and $m_2$ that share a codeword and both 
appear in $S$ with positive probability. Conditioned on their both appearing in 
$S$, both of these messages have equal probability of being the message drawn
from $P$ given the common codeword.
Therefore, whatever Bob chooses to output upon receiving his codeword is wrong
with positive probability, contradicting the assumption that the scheme is
error-free.
\end{proof-of-lemma}

\noindent
Now, since our choice of $\alpha$ ($\log\alpha= \log k+\frac{1}{2}\log\log k+
O(1)$ and $\alpha>k$) ensures that $2\log k-2\log\log k$ is at least $2\log
\alpha-3\log\log\alpha-O(1)$, we find that the redundancy is indeed likewise at
least $2\log\alpha-3\log\log\alpha-O(1)$ since, as shown in Lemma~\ref{hp-lem},
$H(P)$ is also $O(1)$.
\end{proof}

We note that apart from the encoding length, this lower bound is essentially 
tight in another respect: If there are fewer than $\alpha^2$ messages, then 
these can be indexed by using less than $2\log\alpha$ bits, and so clearly a 
better scheme is possible by ignoring the priors entirely and simply indexing 
the messages in this case. Thus, no hard example for uncertain priors coding can
use a substantially smaller set of messages.

Although our hard example uses a source distribution with very low entropy, we
also briefly note that it is possible to extend it to examples of hard source 
distributions with high entropy. Namely, suppose that there is a second, 
independent, high-entropy distribution $T$, and that Alice and Bob wish to solve
two independent source coding problems: one with their $\alpha$-close priors $P$
and $Q$, and the other a standard source coding problem for $T$. The joint 
distributions $PT$ and $QT$ are then $\alpha$-close and have entropy essentially
$H(T)$. The analysis is now similar to the Slepian-Wolf Theorem~\cite{sw73}: The
``side problem'' of source coding for $T$ cannot reduce the coding length for 
uncertain priors coding of $P$ and $Q$ since Alice and Bob can simulate drawing 
an independent message $m''$ from $T$ using the shared randomness; this shared 
message $m''$ can then be used by Bob to decode a hash of the message Alice 
would send in the joint coding problem.

\subsection{The price of uncertainty when errors are allowed}
\label{poserror-lb-sec}

We now turn to the setting where Alice and Bob are allowed to fail at the
communication task with positive probability $\epsilon$. Haramaty and 
Sudan~\cite{hs16} introduced this setting, but only considered deterministic
schemes for this problem. 

\subsubsection{Efficient uncertain priors coding when errors are allowed}
We first note that the original techniques from Braverman and Rao~\cite{br14} 
(used in the original work of Juba et al.~\cite{jkks11}) give a potentially much
better upper bound when errors are allowed:
\begin{theorem}\label{poserror-ub}
For every $\alpha$ and $\epsilon$, there is an uncertain priors compression 
scheme with expected code length $H(P)+\log\alpha+\log\frac{1}{\epsilon}+1$ that
is correct with probability $1-\epsilon$ when $P$ and $Q$ are $\alpha$-close.
(Actually, we only require $Q(x)\geq P(x)/\alpha$ for all $x$.)
\end{theorem}
\begin{proof}
The scheme is a simple variant of the original uncertain priors coding: we use
the shared randomness to choose infinitely long common random strings $r_m$ for 
each message $m$. Then, to encode a message $m$, Alice sends the first $i=
\left\lceil\log\frac{\alpha}{P(m)\epsilon}\right\rceil$ bits of $r_m$. Bob 
decodes this message by computing the set $S$ of messages $m'$ such that the 
first $i$ bits of $r_{m'}$ agree with the codeword, and outputs some $m'$ that 
maximizes $Q(m')$.

To see that $m'=m$ with probability at least $1-\epsilon$, note that it suffices
to show that of the at most $1/Q(m)$ messages $m'$ with probability at least 
$Q(m)$, no $r_{m'}$ is consistent with the first $i$ bits of $r_m$. Now, since
$Q(m)\geq P(m)/\alpha$, Alice has sent at least $\log 1/Q(m)+\log 
1/\epsilon$ bits of $r_m$. The probability that some $r_{m'}$ agrees with so
many bits of $r_m$ is at most $1/2^{\log 1/Q(m)+\log 1/\epsilon}=\epsilon Q(m)$.
Therefore, by a union bound over the $1/Q(m)$ possible high-likelihood messages 
$m'$, the probability that any has $r_{m'}$ consistent with $r_m$ is at most 
$\epsilon$. Thus, none are consistent and decoding is correct with probability
at least $1-\epsilon$.

Finally, we note that the expected codeword length is exactly
\begin{align*}
\bbE_P\left[\left\lceil\log\frac{\alpha}{P(m)\epsilon}\right\rceil\right]&\leq
\bbE_P\left[\log \frac{1}{P(m)}\right]+\log\alpha+\log\frac{1}{\epsilon}+1\\
&=H(P)+\log\alpha+\log 1/\epsilon+1
\end{align*}
as promised.
\end{proof}

\paragraph{Computational aspects}\label{comp-eff}
We observe that the scheme presented in the proof of Theorem~\ref{poserror-ub}
has another desirable property, namely that the encoding is extremely efficient
in the following sense: Given a message $m$, Alice only needs to compute
$P(m)$ (the density function for $P$ at $m$), and only needs to look up the
prefix of $r_m$ of length $\lceil\log \alpha/P(m)\epsilon\rceil$. This is in
contrast to the error-free encoding scheme of Juba et al.~\cite{jkks11}, which
requires Alice to first identify the messages $m'$ with $P(m')\geq P(m)/
\alpha^2$, and then read up to $\log 1/P(m)+2\log\alpha+2$ bits of $r_{m'}$, 
until some index in which $r_m\neq r_{m'}$, for {\em each} such $m'$. Observe 
that this may be as many as $\alpha^2/P(m)-1$ other messages, which may be quite
large (if the individual densities $P(m)$ are small). So, in comparison, the 
positive-error coding is rather computationally efficient. The natural decoding 
procedure, unfortunately, still also requires examining the density of up to 
$1/Q(m)$ messages. We will return to these issues when we discuss open problems.

\subsubsection{Lower bound for uncertain priors coding with errors}
We now give a lower bound showing that the redundancy is essentially $\log\alpha
+\log 1/\epsilon$, up to terms of size $O(\log\log\alpha)$. Thus, in this 
$\epsilon$-error setting, we also identify the ``price'' of uncertainty $\alpha$
up to lower-order terms. The proof is a natural extension of the proof of 
Theorem~\ref{errorless-lb}, showing that when the codes are so short, a 
miscommunication is not only possible but must be likely.

\begin{theorem}\label{poserror-lb}
For every compression scheme for uncertain priors that is correct with 
probability $1-\epsilon$ and every sufficiently large $\alpha>\frac{1}
{\epsilon}$, there is some pair of $\alpha$-close priors for which the scheme
suffers redundancy at least $\log\alpha+\log\frac{1}{\epsilon}-
\frac{9}{2}\log\log\alpha-O(1)$.
\end{theorem}
\begin{proof}
For each $\alpha$, we will consider the same family of priors as used in the 
proof of Theorem~\ref{errorless-lb}: that is, we let $k$ be the largest 
integer such that $k\sqrt{\log k}\leq\alpha$ and consider pairs of priors on a
message set of size $k^2+1$, indexed by a distinguished message $m$ and a set
of $k+1$ messages $S$ that includes $m$; Alice's prior $P$ gives $m$ probability
$1-1/\log k$ and the other $k^2$ messages probability $1/k^2\log k$, and Bob's 
prior gives the messages in $S$ probability $1/k\sqrt{\log k}$, and gives the 
other $k^2-k$ messages uniform probability. Lemma~\ref{pqclose-lem} shows that 
every pair $P$ and $Q$ are $\alpha$-close, and Lemma~\ref{hp-lem} shows that 
$H(P)=O(1)$ for every $P$ in the family. It thus again remains only to show that
for every coding scheme that is correct with probability $1-\epsilon$, there
exists a pair of priors in the family under which the expected code length must
be large.
\begin{lemma}
Any coding scheme that is correct with probability at least $1-\epsilon$ must
have expected codeword length at least $\log k+\log\frac{1}{\epsilon}-4\log\log
k$ for some $\alpha$-close pair $P$ and $Q$ as described above.
\end{lemma}
\begin{proof-of-lemma}
We again note that by the min-max principle, it suffices to consider 
deterministic coding schemes when a pair of priors (given by $m$ and $S$) are 
chosen uniformly at random from our family, and consider any scheme in which the
expected codeword length is at most $\log k+\log\frac{1}{\epsilon}-4\log\log k$.
As before, Markov's inequality guarantees that with probability at least $1
\left/\left(\frac{\log k+\log 1/\epsilon}{2\log\log k}-1\right)\right.$, the 
code length is at most $\log k+\log\frac{1}{\epsilon}-2\log\log k$. Now, 
recalling that we have the uniform distribution over the $k^2+1$ messages, this 
means that there are $\ell\geq\frac{k^2+1}{\frac{\log k+\log 1/\epsilon}{2\log
\log k}-1}$ messages with such short codes. But, now there are at most $2\frac
{k}{\epsilon\log^2 k}$ codewords of this length.

We now consider, for a uniformly chosen message from this set (conditioned on
having such a short codeword), the expected number of messages that are coded by
the same codeword. We will refer to each pair that share a codeword as a
``collision.'' Letting $N_c$ denote the number of these messages coded by
the codeword $c$, this is
\[
\frac{1}{\ell}\sum_x\#\{y\mathrm{\ with\ the\ same\ code\ as\ }x\}
=\frac{1}{\ell}\sum_cN_c^2.
\]
Noting that $\sum_cN_c=\ell$, we know that this expression is minimized when
all $N_c$ have equal size. That is, the expected number of collisions in this
conditional distribution is at least $\frac{\ell\epsilon\log^2 k}{2k}$. Since a
uniformly chosen message hits this conditional distribution with probability
at least $1\left/\left(\frac{\log k+\log 1/\epsilon}{2\log\log k}-1\right)
\right.$, overall we have that the expected number of collisions is at least
\[
\ell\frac{\epsilon \log^2 k}{2k}\frac{1}{\frac{\log k+\log 1/\epsilon}{2\log
\log k}-1}\geq \ell\frac{\epsilon\log k}{4k}.
\]

So, consider a message $m^*$ for Alice that is sampled by first choosing a 
random pair $(P,Q)$ from the family of priors and then sampling from $P$. Note 
that regardless of whether or not $m^*$ is in $S$, there are at least $k$ 
(additional) members of $S$ that are chosen uniformly at random, and each one 
collides with $m^*$ with probability at least 
\[
\frac{N_c}{\ell}\frac{1}{\frac{\log k+\log 1/\epsilon}{\log\log k}-1}, 
\]
which is at least $\frac{\epsilon\log\log k}{8k}$. Thus, for sufficiently large 
$k$, the probability that no member of $S$ shares a code with $m^*$ is less than
\[
\left(1-\frac{\epsilon\log\log k}{8k}\right)^k\leq 
%1-k\frac{\epsilon\log\log k}{8k}+o(1)
1-2\epsilon.
\]
Now, by symmetry of the uniform choice of $m^*$ and the colliding elements of 
$S$, whenever there is at least one collision, Bob's output is wrong with 
probability at least $1/2$. Thus, Bob is wrong in this case with probability 
greater than $\epsilon$, and so the scheme has error greater than $1-\epsilon$.
\end{proof-of-lemma}

\noindent
Since again $\log k\geq\log\alpha-\frac{1}{2}\log\log\alpha-O(1)$ and $\log\log
k\leq\log\log\alpha$, the theorem now follows immediately.
\end{proof}

\section{Suggestions for future work}\label{future-sec}
We now conclude with a handful of natural questions that are unresolved by
our work. First of all, in both the error-free and positive error settings, the
known upper and lower bounds still feature gaps of size $\Theta(\log H(P)+\log
\log\alpha)$. It would be nice to tighten these results further, reducing the 
gap to $\Theta(1)$ if possible in particular (as is known for standard source 
coding). We stress that it may be that {\em neither} the known upper bounds nor 
the known lower bounds are optimal. In particular, since we have considered 
one-shot source coding, we do not need to use prefix codes, and so Kraft's 
inequality does not necessarily apply. It is conceivable that the upper bound 
can be tightened somewhat as a consequence, along the lines of the tightening of
standard source coding achieved by Alon and Orlitsky~\cite{ao94} for example.%
\footnote{We thank Ahmad Beirami for bringing this to our attention.}
We note that there are examples such as uniform distributions under which
the entropy actually matches the optimal encoding length, so the scope for such
improvement is limited to examples such as a geometric distribution in which
there is a significant gap.

Of course, as studied by Shannon~\cite{shannon48}, the prefix-code bounds 
essentially determine the optimal encoding lengths when we wish to transmit more
than one message sampled from the same source. A {\em very} natural second set 
of related questions independently suggested to us by S. Kamath, S. Micali, and 
one of our reviewers concerns the analogous question when multiple messages are 
sampled from the same source, but the sender and receiver are uncertain about
each others' priors. It is clear that if the source is fixed, as the number of
messages grows, the receiver can construct an empirical estimate of the source 
distribution. So in the limit, this obviates the need for schemes such as those 
considered here: standard compression techniques can be used and will achieve 
an optimal coding length (per message) that again approaches the entropy of the 
source. The interesting question concerns what happens when the number of 
messages sent is relatively small, or scales with the uncertainty $\alpha$, or 
relatedly (c.f. our earlier observation that a set of messages of size roughly 
$\alpha^2$ is necessary for our $2\log\alpha$ lower bound to hold), scales with 
the size of the distribution's support. Indeed, after decoding a single message 
(which, with high probability, is the distinguished message $m$), the ``hard'' 
example distributions we constructed subsequently cost $O(1)$ bits per message 
in expectation. One can construct distributions with more ``high-probability'' 
messages in the sender's distribution, but this naturally seems to reduce the 
redundancy since the encodings inherently require more bits. Noting that in this
work we find that redundancy approximately $2\log\alpha$ is necessary when 
$\sim|M|^0=O(1)$ messages are sent and that certainly the redundancy is $O(1)$ 
when $\sim |M|^1$ messages are sent, Kamath~\cite{kamath14} conjectures:
\begin{conjecture}[S. Kamath]
For all $\beta\in [0,1)$, when $\tilde{\Theta}(|M|^\beta)$ messages are sent, 
the optimal per-message redundancy for uncertainty $\alpha$ that is 
$\tilde{O}(|M|^{\frac{1}{2-2\beta}})$ is $(2-2\beta)\log\alpha+\Theta(1)$.
\end{conjecture}
Finally, we note that since these bounds primarily concern small numbers of
messages in a noninteractive setting, questions about the gap between one-to-one
and prefix coding still arise. Analogous to the source coding of a single 
message, it is conceivable that we could achieve similar savings along the lines
obtained by Szpankowski and Verd\'{u}~\cite{sv11} for source coding with 
multiple messages. Similar bounds on the savings for universal compression (with
unknown source distributions) have also been established in works by Kosut and
Sankar~\cite{ks14} and Beirami and Fekri~\cite{bf14}.

Third, we do not have any better lower bounds for the deterministic or 
imperfect-randomness settings, respectively from Haramaty and Sudan~\cite{hs16} 
and Canonne et al.~\cite{cgms14}. The known upper bounds in these settings are 
{\em much} weaker, featuring redundancy that grows at least linearly in the 
entropy of the source distribution, and in the case of the deterministic codes 
of Haramaty and Sudan, some kind of dependence on the size of the distribution's
support. Is this inherently necessary? Haramaty and Sudan give some reasons to 
think so, noting a connection between graph coloring and uncertain priors coding
schemes: namely, they identify geometric distributions with the nodes of a graph
and include edges between all $\alpha$-close geometric distributions. They 
identify the ``colors'' with the messages sent by the high-probability element 
in each distribution. They then point out that a randomized uncertain priors
scheme is essentially a ``fractional coloring'' of this graph, whereas a
deterministic scheme is a standard coloring. Thus, there is reason to suspect
that this problem may be substantially harder, and so a stronger lower bound
(e.g., depending on the size of $M$) may be possible for the deterministic
coding problem.

As for the imperfectly shared randomness setting, Canonne et al.~\cite{cgms14}
give some lower bounds for a communication complexity task (sparse gap inner
product) showing that imperfectly shared randomness may require substantially
more communication than perfectly shared randomness. Although this does not
seem to immediately yield a strong lower bound for the uncertain priors problem
(as it is for a very specific communication problem), their technique may
provide a starting point for a matching lower bound on the redundancy in that 
setting as well.

Finally, none of these works address the question of the {\em computational}
complexity of uncertain priors coding, a question originally raised in the
work of Juba et al.~\cite{jkks11}. In particular, codes such as arithmetic 
coding (first considered by Elias in an unpublished work) can encode a message 
$m$ by making a number of queries to a CDF for the source distribution $P$ that 
is {\em linear in the code length}, and has similar computational complexity. 
Although the code length is not as tight as Huffman coding~\cite {huffman52} for
one-shot coding, it is nearly optimal. The question is whether or not a 
similarly computationally efficient, near-optimal compression scheme for 
uncertain priors coding exists. We observed in Section~\ref{comp-eff} that our
encoding scheme for the positive error setting was both quite efficient, given
the ability to look up the desired prefix of the random encoding efficiently, 
and furthermore was close to optimal in encoding length. But, the natural
decoding strategy is quite inefficient in terms of the number of queries to the
density function required in particular: it requires up to $1/Q(m)$ queries to
decode a message $m$, where $Q$ is the receiver's density function. In the case 
of error-free coding, the decoding strategy is the same (and hence just as
inefficient), and moreover, the known encoding strategies are similarly
inefficient. Thus, the use of either scheme at present is computationally
prohibitive, which is a barrier to their use.

% Generated by IEEEtran.bst, version: 1.13 (2008/09/30)

\end{document}